\documentclass[pdflatex,sn-mathphys]{sn-jnl}

\jyear{2021}%
\usepackage{commath}
\usepackage{caption}
\usepackage{makecell}
\usepackage{tabu}
\usepackage{adjustbox}
\theoremstyle{thmstyleone}%
\newtheorem{theorem}{Theorem}
\theoremstyle{thmstyletwo}%
\newtheorem{example}{Example}%
\newtheorem{remark}{Remark}%

\theoremstyle{thmstylethree}%
\newtheorem{definition}{Definition}%
\newcounter{cases}
\newcounter{subcases}[cases]
\newenvironment{mycases}
  {%
    \setcounter{cases}{0}%
    \setcounter{subcases}{0}%
    \def\case
      {%
        \par\noindent
        \refstepcounter{cases}%
        \textbf{Case \thecases.}
      }%
    \def\subcase
      {%
        \par\noindent
        \refstepcounter{subcases}%
        \textit{Subcase (\thesubcases):}
      }%
  }
  {%
    \par
  }
\renewcommand*\thecases{\arabic{cases}}
\renewcommand*\thesubcases{\roman{subcases}}

\begin{document}

\title[A Direct Construction of 2D-CCC]{A Direct Construction of 2D-CCC with Arbitrary Array Size and Flexible Set Size Using Multivariable Function}


\author*[1]{\fnm{Gobinda} \sur{Ghosh}}\email{sagarghosh798@gmail.com}

\author[2]{\fnm{Sachin} \sur{Pathak}}\email{sachiniitk93@gmail.com}

\affil*[1]{\orgdiv{School of Technology}, \orgname{Woxsen University}, \orgaddress{ \city{Hyderabad}, \postcode{502345}, \state{Telangana}, \country{India}}}

\affil[2]{\orgdiv{Department of Mathematics and Basic Sciences}, \orgname{NIIT University}, \orgaddress{\city{Neemrana}, \postcode{301705}, \state{Rajasthan}, \country{India}}}



\abstract{In this paper, we propose 
 direct construction of two-dimensional complete complementary codes (2D-CCCs) with arbitrary array size and flexible set size using multivariable functions (MVF). We investigate row and column sequences peak to mean envelope ratio (PMEPR) of the proposed construction 2D-CCCs. We have demonstrated that in the case of 2D-CCCs with array and set sizes as power-of-two, the maximum PMEPR for both rows and columns is upper-bounded by a value of two. The proposed construction generalizes many of the existing state-of-the-art sequence designs such as Golay complementary pair (GCP), Golay complementary set (GCS), one-dimensional (1D)-CCC, 2D Golay complementary array set (2D-GCAS), and 2D-CCC. }




\maketitle

\section{Introduction}\label{sec1}
{A}~{pair} of sequences is called  a  Golay
complementary pair (GCP) if the aperiodic
auto-correlation sum (AACS) is zero for any non-zero
time-shift \cite{golay1949multi}. Because of its ideal auto-correlation property, it is used in wireless communication engineering for various purposes \cite{popovic1991synthesis,davis1999peak,wang2014pmepr,spasojevic2001complementary,pezeshki2008doppler}. For example, GCP has been used to reduce the peak-to-mean envelope power ratio (PMEPR) in orthogonal frequency division multiplexing (OFDM)
\cite{popovic1991synthesis,davis1999peak,wang2014pmepr}, to estimate channel \cite{spasojevic2001complementary}, and also use for radar communication \cite{pezeshki2008doppler}. In \cite{tseng1972complementary}, Tseng and Liu generalized the concept of GCPs by introducing Golay complementary sets (GCSs) and mutually
orthogonal Golay complementary sets (MOGCS). GCS refers to a set of sequences of the same lengths having zero AACS for all non-zero time shifts. If the set size of GCS is $2$, it is called GCP.
The PMEPR of  GCS-based orthogonal frequency-division multiplexing (OFDM) is limited by the set size of GCS \cite{davis1999peak,schmidt2007complementary,chen2006complementary,paterson2000generalized}. 
On the other hand, MOGCS is a collection of GCSs in which each GCS is orthogonal to each other in terms of their zero cross-correlation sums for all time shifts.
When the set size of MOGCS reaches its upper bound, it is called 
complete complementary code (CCC), first introduced by Suehiro and Hatori in  \cite{suehiro1988n}. 
There are many constructions of CCCs \cite{rathinakumar2008complete,das2019new,das2018novel,sarkar2021multivariable,das2017multiplier,de2007modular} that have a wide range of applications in wireless communications \cite{chen2001multicarrier,liu2014fractional}, information hiding \cite{kojima2013audio,kojima2014disaster}, etc.\\
\text{~~}The sequences discussed above are one-dimensional (1D) sequences. Recently, Liu \textit{et al}.  constructed two-dimensional (2D)-Golay complementary array
sets (GCAS) and demonstrated their applications in massive multiple input multiple output (MIMO) systems equipped with a uniform rectangular array (URA)  \cite{liu2022constructions}. There are many constructions of GCAS available in \cite{liu2022constructions,pai2022two,wang2021new,wang2020constructions}.  
 The authors in \cite{1267872} implemented two-dimensional mutually orthogonal complementary array set (2D-MOCASs) in an ultra-wideband (UWB) multi-carrier code-division multiple access (MC-CDMA) system to deliver a low  PMEPR value in addition to good bandwidth efficient transmission. A 2D-MOCAS becomes a 2D-CCC when the number of component arrays is equal to the set size \cite{wang2021new}. One of the most important uses of 2D-CCC and 2D-GCAS are in omnidirectional precoding (OP) based massive (MIMO) systems \cite{meng2015omnidirectional,liu2022constructions,pai2022designing,zhao2023two,lu2020omnidirectional}. 
Other than this, 2D-CCC has other applications such as 
 multiple access interference (MAI)-free transmission \cite{turcsany2006performance}, 
 image change and motion detection   \cite{el2004new}, 
 and application in 2D-MC-CDMA system\cite{farkas2003two}.
\par even though 2D-CCC has several essential applications, their direct constructions, which do not depend on initial sequences and matrices, are still not available in the literature. Most of the 2D-CCCs are constructed using different 1D or 2D sequences with certain operations
\cite{davidekova2013generalized,turcsany2003two,das2020two}. In \cite{davidekova2013generalized}, 2D-CCC with array size $2^{m}M \times 2^{m}M$ and set size $M^{2}$ is constructed for $m\geq 1$ and $M\geq 2$. In \cite{farkas2003two}, 2D-CCC with array size $M^{2}\times K^{2}$ and set size $MK$ is constructed for $M,K\geq 2$. In \cite{das2020two}, 2D-CCC with array size $M\times M$ and set size $M$ is constructed for $M\geq 2$. All the above constructions depend heavily on some choices of initial sequences and matrices.
 Pai and Chen in  \cite{pai2022designing} used GBF to construct 2D-CCC of set size $2^{k}$ and  array size $2^{n}\times2^{m}$ where $m,n,k\geq 1$ and $m,n\geq k$. To the best of our knowledge, there is no direct construction of 2D-CCC with all possible array sizes and flexible set sizes available in the existing literature.
\par The limitations in array and set sizes of 2D-CCC found in direct constructions within existing literature motivate us to explore the use of multivariable functions (MVF) to achieve greater flexibility in both array and set dimensions. Motivated by the scarcity of 2D-CCC with flexible parameters, in this study, we 
construct  2D-CCC of all possible array sizes of the form $m\times n$ where $m=\prod_{i=1}^{a}p_{i}^{m_{i}}$ $n=\prod_{j=1}^{b}q_{j}^{n_{j}}$, and flexible set size of the form $\prod_{i=1}^{a}p^{k_{i}}_{i}\prod_{j=1}^{b}q^{l_{j}}_{j}$ where $p_{i},q_{j}$ are prime numbers, and $m_{i},n_{j},k_{i},l_{j}$are arbitrary positive numbers by MVF. Since 2D-CCC is a collection of 2D-GCAS, the proposed construction also produces 2D-GCAS with arbitrary array size and flexible set size which has not been reported before either by indirect or direct construction.
The row and column
sequence PMEPR of the 2D-CCC array has been derived which is also not been reported before. The proposed construction has the following generalization abilities: 
\begin{itemize}
    \item The proposed construction provides 2D-GCAS with  array size $p^{m}\times q^{n}$ and set size $p^{k_{1}}q^{k_{2}}$ where $m,n,k_{1}$ and $k_{2}$ are positive integers which generalizes 2D-GCAS given in \cite{liu2022constructions},
    \item The proposed construction yields 2D-CCC with array size $M^{2}\times K^{2}$ and set size $MK$ where $M,K\geq 2$, therefore, the proposed work generalizes the 2D-CCC given in  \cite{farkas2003two}.
    \item The proposed construction  produces 1D-CCC with sequence length $2^{m}$ and set size  $2^{k}$, where $k,m \geq 1$. As a result, the 1D-CCC constructions in \cite{chen2008complete, rathinakumar2008complete} have appeared as a special case of the proposed work.
    \item The proposed construction produces GCAS with array size $2^{m}\times 2^{n}$ and set size $2^{k}$ where $m,n,k\geq 1$. Therefore, the GCAS proposed in \cite{pai2022two,pai2021designing} can be constructed by the proposed method.
    \item The construction of 1D-CCC in  \cite{sarkar2021multivariable} with various lengths is covered by the proposed construction,
\item The proposed method presents $2^{k}$ GCAS with array size $2^{kr}\times2^{ks}$ and set size $2^{k}$ where $k,r,s$ are positive integer.  Therefore, the GCAS proposed in \nocite{wang2021new} \cite[Th.7]{wang2021new12} appeared as a special case of our proposed construction.
\end{itemize}

\par The rest of the paper is structured as follows.
Section \ref{SEC2} provides useful definitions. Section \ref{SEC3} describes 2D-CCC construction. Section \ref{SEC4} examines the PMEPR of row and column sequences in 2D-CCC arrays and provides generalizations of the proposed 2D-CCC. Section \ref{SEC5} compares the proposed 2D-CCC to the current state-of-the-art. Section \ref{SEC6} concludes.

\section{Preliminaries}\label{SEC2}	
\subsection{One-Dimensional Sequences}
	\begin{definition}\cite{ghosh2023construction,shen2023new}
		For two complex-valued sequences $\mathbf{a}= (a_0, a_1, \dots, a_{l-1})$ and $\mathbf{b} = (b_0, b_1, \dots, b_{l-1})$ of length $l$, the aperiodic cross-correlation sum (ACCS) of sequence at the time shift $u$ is defined by
		\begin{equation}
			\boldsymbol{C}(\mathbf{a}, \mathbf{b}) (u) =
			\begin{cases}
				\sum\limits_{i=0}^{l-u-1} a_i b_{i+u}^*,  & 0\leq u < l ;\\
				\sum\limits_{i=0}^{l+u-1} a_{i-u} b_{i}^*, & -l < u \leq -1 ;
			\end{cases}
		\end{equation}
		where $(\cdot)^*$ denotes the complex conjugate. If $\mathbf{a} = \mathbf{b}$, then the function is called AACS and is denoted by $\boldsymbol{C}(\mathbf{a})(u)$.
	\end{definition}
	\begin{definition}\cite{golay1961complementary}
 Two sequences $\mathbf{a}, \mathbf{b}$ each of length $l$ is called Golay complementary pair (GCP) if
		\begin{equation}
			 \boldsymbol{C}(\mathbf{a})(u)+ \boldsymbol{C}(\mathbf{b}) (u) =
			\begin{cases}
				2l,  & u=0 ;\\
				0, & \text{otherwise}.
			\end{cases}
		\end{equation}
	\end{definition}
	\begin{definition}[Code (\cite{ghosh2022direct})]
		An ordered set of sequences $\mathcal{A}=\{\mathbf{a}_0, \mathbf{a}_1, \dots, \mathbf{a}_{r-1}\}$ of the same length $l$, is called a code. The code $\mathcal{A}$ can be written in a matrix form of order $r\times l$ as well, where $i$th row is given by the sequence $\mathbf{a}_{i}$ where $0\leq i \leq r-1$.
	\end{definition}

	\begin{definition}[ACCS of code (\cite{ghosh2022direct})]
		Let $\mathcal{A}_{1}=\{\mathbf{a}_0^1, \mathbf{a}_1^1, \dots, \mathbf{a}_{r-1}^{1}\}$, $ \mathcal{A}_{2}=\{\mathbf{a}_0^2, \mathbf{a}_1^2, \dots, \mathbf{a}_{r-1}^2\}$ be two codes, where $\mathbf{a}_{i}^{1}$ and $\mathbf{a}_{i}^{2}$ are sequences of length $l$. The ACCS of  $\mathcal{A}_{1}$ and $\mathcal{A}_{2}$ is defined by
		\begin{equation} 
			\boldsymbol{C}(\mathcal{A}_{1},\mathcal{A}_{2})(u)=
			\displaystyle \sum _{i=0}^{r-1} \boldsymbol{C}(\mathbf {a}_i ^{1},\mathbf {a}_i ^{2})(u).
		\end{equation}
	\end{definition}

	\begin{definition}[MOGCS(\cite{ghosh2022direct})]
		A collection $\mathcal{X}= \{\mathcal{A}_0, \mathcal{A}_1,\ldots, \mathcal{A}_{k-1}\}$  of codes each of size $r\times l$ is called mutually orthogonal Golay complementary sets (MOGCS) if it satisfies
		\begin{equation} 
			\boldsymbol{C}(\mathcal{A}_{\nu_{1}},\mathcal{A}_{\nu_{2}})(u) =
			\begin{cases}
				rl, & u =0,\quad \nu_{1}=\nu_{2};\\
				0, & 0 < |u | < l, \quad \nu_{1}=\nu_{2};\\
				0, & |u | < l,\quad \nu_{1}\neq \nu_{2}. 
			\end{cases}
		\end{equation}
		When $k=r$, then $\mathcal{X}$ is called CCC.
	\end{definition}
\subsection{Two-Dimensional Array}	
\begin{definition}[\cite{das2020two}]
Let $\mathbf{A}=\left(a_{i,j}\right)$ and $\mathbf{B}=\left(b_{i,j}\right)$ be a complex valued array of size $l_{1} \times l_{2}$ where $0 \leq i<l_{1},0 \leq j<l_{2}$.
The 2D-ACCS of arrays $\mathbf{A}$ and $\mathbf{B}$ at shift $\left(u_{1}, u_{2}\right)$ is defined as
   \begin{equation*}
   \begin{split}
       &\boldsymbol{C}\left(\mathbf{A},\mathbf{B}\right)\left(u_{1}, u_{2}\right)=\\
       &\begin{cases}
\sum_{i=0}^{l_{1}-1-u_{1}}\sum_{j=0}^{l_{2}-1-u_{2}} a_{i, j}b^{*}_{i+u_{1}, j+u_{2}},~~~\text{if}~~0 \leq u_{1}<l_{1},0 \leq u_{2}<l_{2};\\
\sum_{i=0}^{l_{1}-1-u_{1}}\sum_{j=0}^{l_{2}-1+u_{2}} a_{i, j-u_{2}}b^{*}_{i+u_{1}, j},~~~\text{if}~~0 \leq u_{1}<l_{1}, -l_{2} < u_{2}<0;\\
\sum_{i=0}^{l_{1}-1+u_{1}}\sum_{j=0}^{l_{2}-1-u_{2}} a_{i-u_{1}, j}b^{*}_{i, j+u_{2}},~~~\text{if}~-l_{1} < u_{1}<0,0 \leq u_{2}<l_{2};\\
\sum_{i=0}^{l_{1}-1+u_{1}}\sum_{j=0}^{l_{2}-1+u_{2}} a_{i-u_{1},j-u_{2}}b^{*}_{i,j},~~~\text{if}~-l_{1} < u_{1}<0,-l_{2} < u_{2}<0;
\end{cases}
   \end{split}
\end{equation*}
\end{definition}
If $\mathbf{A}=\mathbf{B}, \boldsymbol{C}\left(\mathbf{A}, \mathbf{B}\right)\left(u_{1},u_{2}\right)$ is called the 2D-AACS of $\mathbf{A}$ and referred to as $\boldsymbol{C}\left(\mathbf{A}\right)\left(u_{1},u_{2}\right)$.
\begin{definition}\cite{liu2022constructions}
Let $s\geq2$ be an integer, and define the set $\mathcal{N}=\{\mathbf{A}_{0},\mathbf{A}_{1},\ldots,\mathbf{A}_{s-1}\}$, where $\mathbf{A}_{i}$'s are arrays of  size $l_{1}\times l_{2}$ where $1\leq i\leq s-1$. We refer to the set $\mathcal{N}$ as a two-dimensional Golay complementary array set (2D-GCAS) if
	\begin{equation}
			 \sum_{i=0}^{s-1}\boldsymbol{C}(\mathbf{A}_{i})(u_{1},u_{2})=
			\begin{cases}
				sl_{1}l_{2},  & (u_{1},u_{2})=(0,0) ;\\
				0, & \text{otherwise}.
			\end{cases}
		\end{equation}
	\end{definition}

\begin{definition}[\cite{das2020two}]
Consider the set $\mathcal{M}=\left\{\mathbf{A}^{k} \mid k=\right.$ $0,1, \ldots, s-1\}$, where each set $\mathbf{A}^{k}=\left\{\mathbf{A}_{0}^{k}, \mathbf{A}_{1}^{k}, \ldots, \mathbf{A}_{s-1}^{k}\right\}$ is composed of $s$ arrays of size $l_{1} \times l_{2}$, i.e., each $\mathbf{A}_{i}^{k}$ is an array of size $l_{1}\times l_{2}$ . The set $\mathcal{M}$ is said to be $(s,s,l_{1},l_{2})$ 2D-CCC if the following holds
\begin{equation}
\begin{split}  
&\boldsymbol{C}\left(\mathbf{A}^{k},\mathbf{A}^{k^{\prime}}\right)\left(u_{1},u_{2}\right)\\
&=\sum_{i=0}^{s-1} \boldsymbol{C}\left(\mathbf{A}_{i}^{k}, \mathbf{A}_{i}^{k^{\prime}}\right)\left(u_{1}, u_{2}\right)\\
&= \begin{cases}sl_{1}l_{2}, \quad\left(u_{1}, u_{2}\right)=(0,0), k=k^{\prime} ;\\
0, \quad\left(u_{1}, u_{2}\right)\neq(0,0), k=k^{\prime} ;\\
0,     ~~~~k\neq k^{\prime}.
\end{cases}
\end{split}
\end{equation}.
\end{definition}  
\subsection{Multivariable Function~\cite{ghosh2023optimal2dzcacs}}

Consider two positive integers $a$ and $b$ and consider $i$ and $j$ such that $1\leq i\leq a$ and $1\leq j\leq b$. For every value of $i$ and $j$ within these ranges, associate a positive integer, denoted $m_i$ and $n_j$ respectively. Additionally, associate a prime number, denoted $p_{i}$ and $q_{j}$ respectively. Now, consider the sets $\mathbb{A}_{p_{i}}=\{0,1,\ldots,p_{i}-1\}$ and $\mathbb{A}_{q_{j}}=\{0,1,\ldots,q_j-1\}$, which are subsets of the ring of integers $\mathbb{Z}$ under usual addition and multiplication.
	A multi-variable function (MVF) can be defined with a domain constructed from cartesian products of the power sets of $\mathbb{A}_{p_i}$ and $\mathbb{A}_{q_j}$, and a range in the ring of integers modulo a positive integer $\lambda$, i.e., $\mathbb{Z}_\lambda=\{0,1,\ldots,\lambda-1\}$ as
	$$f: \left(\mathbb{A}_{p_1}^{m_1}\!\!\times \!\mathbb{A}_{p_2}^{m_2}\!\times \dots \times \mathbb{A}_{p_a}^{m_a}\right)\!\times\!\left(\mathbb{A}_{q_1}^{n_1}\!\!\times \mathbb{A}_{q_2}^{n_2} \times \dots \times \mathbb{A}_{q_b}^{n_b}\right)\!\!\rightarrow\!\! \mathbb{Z}_\lambda,$$
	 Next, we consider two integers $c$ and $d$ in the range of  $0\leq c \leq p_{1}^{m_{1}}p_{2}^{m_{2}}\ldots p_{a}^{m_{a}}-1$ and  $0\leq d\leq q_{1}^{n_{1}}q_{2}^{n_{2}}\ldots q_{b}^{n_{b}}-1$.  The integers $c$ and $d$ can be decomposed as per the following equations:
 \begin{equation}
 \label{c,d}
    \begin{split}
&c=c_{1}+c_{2}p_{1}^{m_{1}}+\ldots+c_{a}p_{1}^{m_{1}}p_{2}^{m_{2}}\ldots p_{a-1}^{m_{a-1}},\\
&d=d_{1}+d_{2}q_{1}^{n_{1}}+\ldots+d_{b}q_{1}^{n_{1}}q_{2}^{n_{2}}\ldots q_{b-1}^{n_{b-1}},
    \end{split}
\end{equation}
where $0\leq c_{i}\leq p_{i}^{m_{i}}-1$ and $0\leq d_{j}\leq q_{j}^{n_{j}}-1$. 
Let $\mathbf{C}_{i}=(c_{i,1},c_{i,2},\ldots,c_{i,m_{i}})\in \mathbb{A}_{p_{i}}^{m_{i}}$
be the vector representation of $c_{i}$ with base $p_{i}$, i.e., $c_{i}=\sum_{k=1}^{m_{i}}c_{i,k}p_{i}^{k-1}$ and $\mathbf{D}_{j}=(d_{j,1},d_{j,2},\ldots,d_{j,n_{j}})\in \mathbb{A}_{q_{j}}^{n_{j}}$ be the vector representation of  $d_{j}$  with base $q_{j}$,  i.e., $d_{j}=\sum_{l=1}^{n_{j}}d_{j,l}q_{j}^{l-1}$ where $1\leq i\leq a$, $1\leq j\leq b, 0\leq c_{i,k}\leq p_{i}-1,0\leq d_{j,l}\leq q_{j}-1 $. We define vectors associated with $c$ and $d$ as $$\phi(c)=\left(\mathbf{C}_{1},\mathbf{C}_{2},\ldots,\mathbf{C}_{a}\right)$$ and $$\phi(d)=\left(\mathbf{D}_{1},\mathbf{D}_{2},\ldots,\mathbf{D}_{b}\right),$$ respectively.
We define an array associated with $f$ as $\psi({f})$, where
\begin{equation}
    \psi({f})=\left(\begin{array}{cccc}
\omega_{\lambda}^{f_{0,0}} & \omega_{\lambda}^{f_{0,1}} & \cdots & \omega_{\lambda}^{f_{0,m}} \\
\omega_{\lambda}^{f_{1,0}} & \omega_{\lambda}^{f_{1,1}} & \cdots & \omega_{\lambda}^{f_{1,m}} \\
\vdots & \vdots & \ddots & \vdots \\
\omega_{\lambda}^{f_{n,0}} & \omega_{\lambda}^{f_{n,1}} & \cdots & \omega_{\lambda}^{f_{n,m}}
\end{array}\right)
\end{equation}
 $m=p_{1}^{m_{1}}p_{2}^{m_{2}}\ldots p_{a}^{m_{a}}-1,n=q_{1}^{n_{1}}q_{2}^{n_{2}}\ldots q_{b}^{n_{b}}-1$ and the $(c,d)$-th component of $\psi({f})$ is given by $\omega_{\lambda}^{f_{c,d}}$, $\omega_{\lambda}=\exp\left(2\pi\sqrt{-1}/\lambda\right)$ and
$f_{c,d}=f\left(\phi(c),\phi(d)\right)$. For simplicity, from from hereon, we  denote $\left(\mathbb{A}_{p_1}^{m_1}\!\!\times \!\mathbb{A}_{p_2}^{m_2}\!\times \dots \times \mathbb{A}_{p_a}^{m_a}\right)\!\times\!$$(\mathbb{A}_{q_1}^{n_1}\!\!\times \mathbb{A}_{q_2}^{n_2} \times \dots \times \mathbb{A}_{q_b}^{n_b})$ by $\mathcal{C}$.
\subsection{Peak-to-Mean Envelope Power Ratio (PMEPR)
\cite{pai2022two}}
Let $\lambda\geq 2$ and $\mathbb{Z}_{\lambda}$ denotes ring of integer modulo $\lambda$. Let $\mathbf{x}$ be a $\mathbb{Z}_{\lambda}$ valued vector of length $l$ given by $\mathbf{x}=\left(x_{0}, x_{1}, \ldots, x_{l-1}\right)$  
. For the vector $\mathbf{x}$, the OFDM signal is the real part of the complex envelope
$$
S_{\mathbf{x}}(t)=\sum_{i=1}^{l-1} \omega_{\lambda}^{x_{ i}+\lambda f_{i} t},
$$
where $f_{i}=f+i \Delta f, f$ is a constant, $\Delta f$ is an integer multiple of OFDM symbol rate, and $0 \leq \Delta f t \leq 1$. 
The PMEPR of the vector $\mathbf{x}$ is given by
$$
PMEPR\left(\mathbf{x}\right)=\sup _{0 \leq \Delta f t \leq 1} \frac{\abs{S_{\mathbf{x}}(t)}^{2}}{l}.
$$  The term $\frac{\abs{S_{\mathbf{x}}(t)}^{2}}{l}$ is called the instantaneous-to-average power ratio (IAPR).
\section{Proposed Construction of 2D-CCC }\label{SEC3}
In this Section, we construct 2D-CCC and give an example. Let  $0\leq \gamma\leq p_{1}^{m_{1}}p_{2}^{m_{2}}\ldots p_{a}^{m_{a}}-1$ and $0\leq \mu\leq q_{1}^{n_{1}}q_{2}^{n_{2}}\ldots q_{b}^{n_{b}}-1$ be such that
\begin{equation*}
    \begin{split}
    &\gamma=\gamma_{1}+\sum_{i=2}^{a}\gamma_{i}\left(\prod_{i_{1}=1}^{i-1}p_{i_{1}}^{m_{i_{1}}}\right),\\
     &\mu\!=\!\mu_{1}+\sum_{j=2}^{b}\mu_{j}\!\!\left(\prod_{j_{1}=1}^{j-1}q_{j_{1}}^{n_{j_{1}}}\!\right),
    \end{split}
\end{equation*}
where $0\leq \gamma_{i}\leq p_{i}^{m_{i}}-1$ and $0\leq \mu_{j}\leq q_{j}^{n_{j}}-1$, $p_i$ is a prime number or $1$ and $q_j$ is a prime number. Let $\boldsymbol{\gamma}_{i}=(\gamma_{i,1},\gamma_{i,2},\ldots,\gamma_{i,m_{i}})\in \mathbb{A}_{p_{i}}^{m_{i}}$ be the vector representation of $\gamma_{i}$ with base $p_{i}$, i.e, $\gamma_{i}=\sum_{k=1}^{m_{i}}\gamma_{i,k}p_{i}^{k-1}$. Similarly,
$\boldsymbol{\mu}_{j}=(\mu_{j,1},\mu_{j,2},\ldots,\mu_{j,n_{j}})\in \mathbb{A}_{q_{j}}^{n_{j}}$ is the vector representation of $\mu_{j}$ with base $q_{j}$. Let $$\phi(\gamma)=\left(\boldsymbol{\gamma}_{1},\boldsymbol{\gamma}_{2},\ldots,\boldsymbol{\gamma}_{a}\right)\in \mathbb{A}_{p_1}^{m_1}\!\!\times \!\mathbb{A}_{p_2}^{m_2}\!\times \dots \times \mathbb{A}_{p_a}^{m_a}$$ be the vector associated with $\gamma$ and $$\phi(\mu)=\left(\boldsymbol{\mu}_{1},\boldsymbol{\mu}_{2},\ldots,\boldsymbol{\mu}_{b}\right)\in \mathbb{A}_{q_1}^{n_1}\!\!\times \!\mathbb{A}_{q_2}^{n_2}\!\times \dots \times \mathbb{A}_{q_b}^{n_b}$$ be the vector associated with $\mu$.
Let
$\pi_{i}$ and $\sigma_{j}$ be any permutations of the set  $\{1,2,\ldots,m_{i}\}$ and  $\{1,2,\ldots,n_{j}\}$, respectively.
Let us also define a MVF $f:\mathcal{C}\rightarrow\mathbb{Z}_\lambda,$ as 
\begin{equation*}
\label{kanha}
    \begin{split}
           &f\big(\phi(\gamma),\phi(\mu)\big)=f\big((\boldsymbol{\gamma}_{{1}},\boldsymbol{\gamma}_{{2}}, \ldots, \boldsymbol{\gamma}_{{a}}),(\boldsymbol{\mu}_{{1}},\boldsymbol{\mu}_{{2}}, \ldots, \boldsymbol{\mu}_{b})\big)\\
&=\!\!\sum_{i=1}^{a}\!\frac{\lambda}{p_{i}}\!\!\sum_{e=1}^{m_{i}-1}\!\!\gamma_{i, \pi_{i}(e)} \gamma_{i, \pi_{i}(e+1)}+\!\!\sum_{j=1}^{b}\!\frac{\lambda}{q_{j}}\!\!\sum_{o=1}^{n_{j}-1} \mu_{j, \sigma_{j}(o)} \mu_{j, \sigma_{j}(o+1)}\\
&+\sum_{i=1}^{a}\!\sum_{e=1}^{m_{i}}\!d_{i,e} \gamma_{i, e}\!\!+\!\!\sum_{j=1}^{b}\!\sum_{o=1}^{n_{j}} c_{j,o} \mu_{j,o}~(\text{mod}~\lambda)
    \end{split}
\end{equation*}
  \vspace{-0.05cm}where $\lambda$ is divisible by l.c.m$(p_{1},\ldots,p_{a},q_{1},\ldots,q_{b})$, $d_{i,e},c_{j,o}\in \mathbb{A}_{\lambda}$ and 
  all the operations
between those integers, the modulus is performed. From now hereon, we will exclude the term $(\text{mod}~\lambda)$ from the function. Let us define sets $\Theta$ and $T$ as
   \begin{equation*}
  \label{kaka}
      \begin{split}
       &\Theta=\{\theta:\theta=(r_{{1}}, r_{{2}}, \ldots, r_{{a}},s_{{1}}, s_{{2}}, \ldots, s_{{b}})\},\\
       &T=\{t:t=(x_{{1}}, x_{{2}}, \ldots, x_{{a}},y_{{1}}, y_{{2}}, \ldots, y_{{b}})\},
      \end{split}
  \end{equation*}
  where $0\leq r_{i},x_{i}< p_{i}^{k_{i}}$ and $0\leq s_{j},y_{j}< q_{j}^{l_{j}}$ and $k_{i},l_{j}$ are positive integers.
  Again, we define a function 	$a^{\theta}_{t}\!\!:\mathcal{C}\rightarrow\!\! \mathbb{Z}_\lambda,$ as
  \vspace{-0.3cm}
  \begin{equation}{\label{5}}
 \begin{split}
 &a^{\theta}_{t}\big(\phi(\gamma),\phi(\mu)\big)=a^{\theta}_{t}\big((\boldsymbol{\gamma}_{{1}},\boldsymbol{\gamma}_{{2}}, \ldots, \boldsymbol{\gamma}_{{a}}),(\boldsymbol{\mu}_{{1}},\boldsymbol{\mu}_{{2}}, \ldots, \boldsymbol{\mu}_{b})\big)\\
     &=\!\!f\big(\phi(\gamma),\phi(\mu)\big)\!\!+\!\!\sum_{i=1}^{a} \frac{\lambda}{p_{i}} \gamma_{i, \pi_{i}(1)}{r_{i}}+\!\sum_{j=1}^{b} \frac{\lambda}{q_{j}} \mu_{j, \sigma_{j}(1)}{s_{j}}\\
    &+\sum_{i=1}^{a} \frac{\lambda}{p_{i}} \gamma_{i, \pi_{i}(m_{i})}{x_{i}}+\sum_{j=1}^{b} \frac{\lambda}{q_{j}} \mu_{j, \sigma_{j}(n_{j})}{y_{j}},
 \end{split}
\end{equation}
where $d_{\theta}\in \mathbb{A_{\lambda}}$ be any integer for each value of $\theta$. For simplicity, we denote $a^{\theta}_{t}\big(\phi({\gamma}),\phi({\mu})\big)$ by $(a^{\theta}_{t})_{\gamma,\mu}$ and $f\big(\phi({\gamma}),\phi({\mu})\big)$ by $f_{\gamma,\mu}$.
          

     
\!\!\!\!\!
\begin{theorem}\label{KB}
 We define $G_{t}=\{a^{\theta}_{t}:\theta\in \Theta\}$ and $\psi(G_{t})=\{\psi\left(a^{\theta}_{t}\right):\theta\in \Theta\}$. Then the set  $\{\psi(G_{t}):t\in T\}$ forms a 2D
$
(\alpha,\alpha,m,n)-CCC
$, where, $\alpha=\prod_{i=1}^{a}p^{k_{i}}_{i}\prod_{j=1}^{b}q^{l_{j}}_{j}$, $m=\prod_{i=1}^{a}p_{i}^{m_{i}}$ and $n=\prod_{j=1}^{b}q_{j}^{n_{j}}$.
\end{theorem}
\begin{proof}
    The proof of this theorem is given in Appendix \ref{kn}
\end{proof}
\begin{example}
Let $a=1,b=1,p_{1}=2,q_1=3,m_1=2,k_{1}=1,l_{1}=1$, and $n_1=2$. Consider $\pi_1$ and $\sigma_1$ are permutations of the set $\{1,2\}$  such that $\pi_{1}(1)=2, \pi_{1}(2)=1$ and  $\sigma_{1}(1)=2, \sigma_{1}(2)=1$.
Let $\boldsymbol{\gamma}_{1}=(\gamma_{11},\gamma_{12})$ be the vector associated with ${\gamma}_{1}$ where $0\leq \gamma_{1}\leq 3$  and $\boldsymbol{\mu}_{1}=(\mu_{11},\mu_{12})$ be the vector associated with $\mu_{1}$ where $0\leq\mu_{1}\leq 8$.
We define a MVF $f:\mathbb{A}_{2}^{2}\times \mathbb{A}_{3}^{2}\rightarrow \mathbb{Z}_{6}$ as
\begin{equation*}
    \begin{split}
        f\left(\boldsymbol{\gamma}_{1},\boldsymbol{\mu}_{1}\right)=&3\gamma_{1,2}\gamma_{1,1}+\gamma_{1,1}+2\gamma_{1,2}+2\mu_{1,2}\mu_{1,1}+2\mu_{1,1}+\mu_{1,2}.
    \end{split}
\end{equation*}
Consider $\Theta=\{\theta:\theta=(r_1,s_1)\}$, $T=\{t:t=(x_1,y_1)\}$ where $0\leq r_1,x_1\leq 1$,$0\leq s_1,y_1\leq 2$.
Let $c_{\theta}$ be any element in $\mathbb{A}_{6}$, then by the defination of $a_{t}^{\theta}$, we have
\begin{equation*}
\begin{split}
    a_{t}^{\theta}\left(\boldsymbol{\gamma}_{1},\boldsymbol{\mu}_{1}\right)=
    f\left(\boldsymbol{\gamma}_{1},\boldsymbol{\mu}_{1}\right)+3\gamma_{1,2}r_1+3\gamma_{1,1}x_1+2\mu_{1,2}s_1+2\mu_{1\sigma_{1,1}}y_1+c_{\theta}.
\end{split}
\end{equation*}
 Consider a set $G_{t}=\{a^{\theta}_{t}:\theta\in \Theta\}$, then the set  $\mathcal{G}=\{G_{t}:t\in T\}$ forms a 2D 
$
(6,6,4,9)-CCC$. The AACS of any code set from $\mathcal{G}$  is given in  Fig.\ref{fig:my_label3} and the ACCS between any two code sets from $\mathcal{G}$ is given in Fig. \ref{fig:my_label2}. 
\end{example}
\begin{figure}
    \centering
    \includegraphics[width=.6 \textwidth]{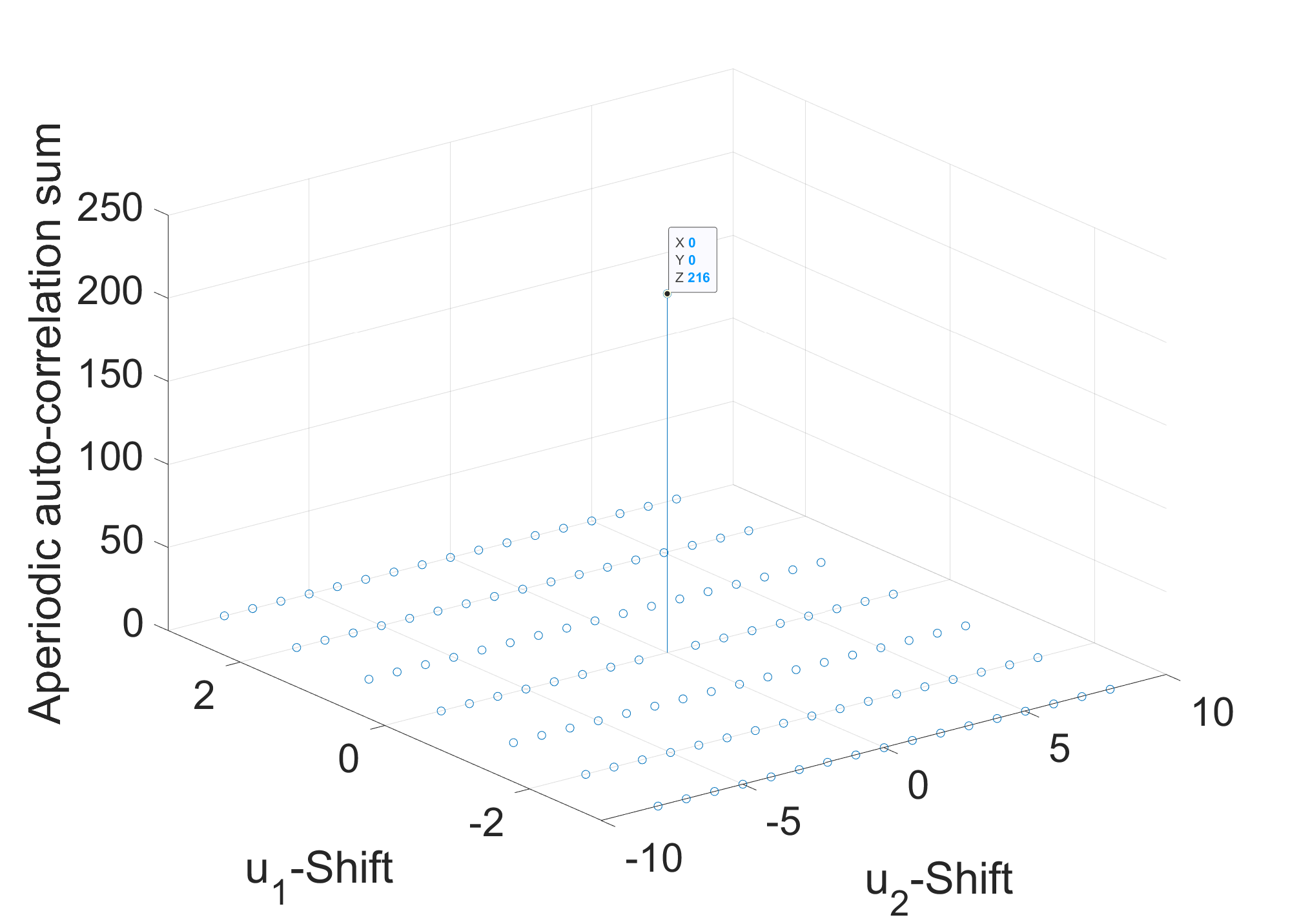}
    \caption{Auto-correlation result of any set of array from $\mathcal{G},$}
    \label{fig:my_label3}
\end{figure}
\begin{figure}[H]
    \centering
    \includegraphics[width=.6 \textwidth]{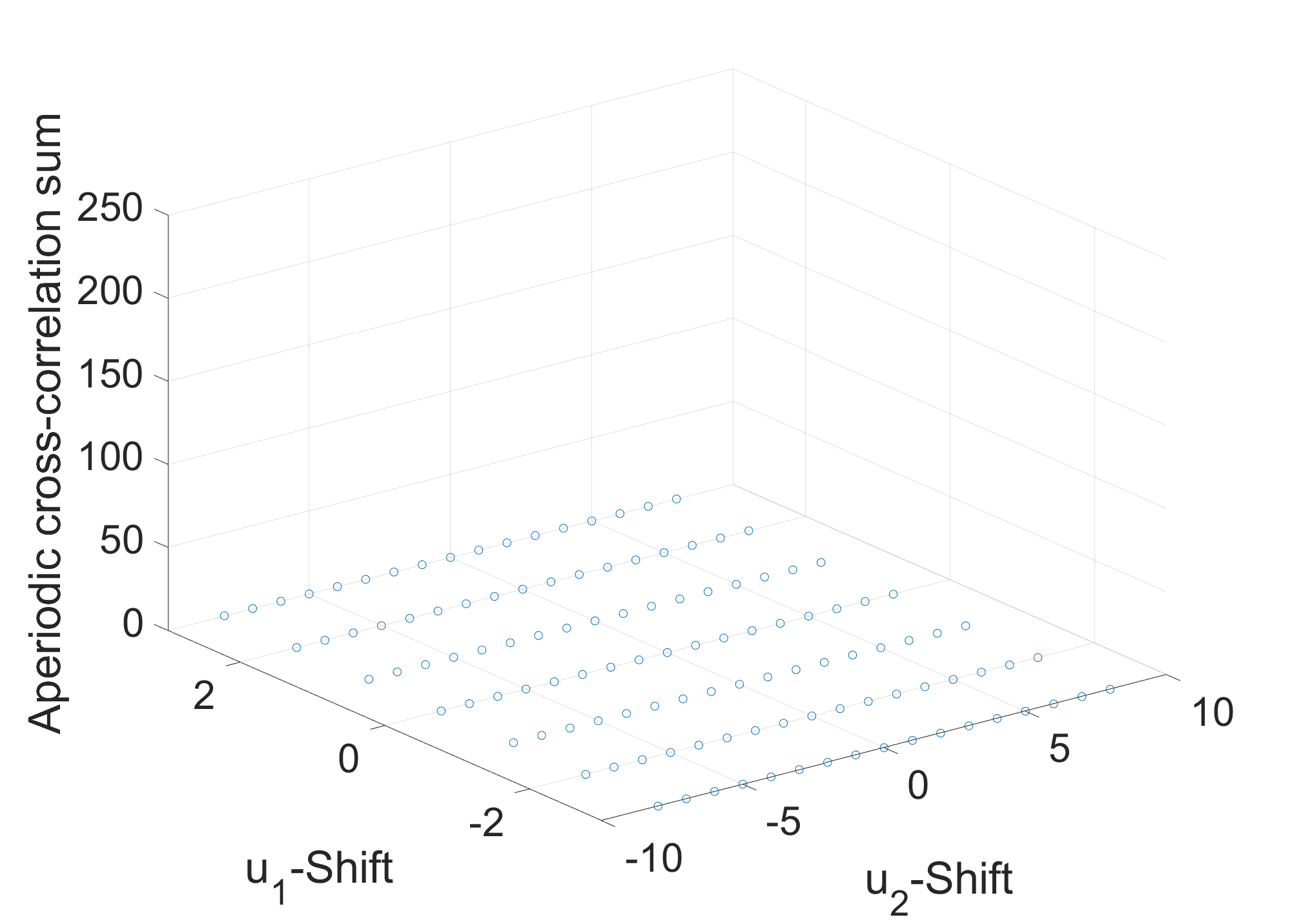}
    \caption{Cross-correlation result of any two sets of array from $\mathcal{G}.$}
    \label{fig:my_label2}
\end{figure}

\section{Bound for PMEPR}\label{SEC4}
In this section we  derive the PMEPR bound of the $2$D-CCC arrays given in \textit{Theorem 1}.\\
\begin{theorem}\label{PMEPRBound}
Given the 2D-CCC arrays as defined in \textit{Theorem \ref{KB}}, the PMEPR of the arrays can be characterized as follows: the PMEPR of the column sequences is bounded by $\max \{p_{i}^{k_{i}}:1\leq i\leq a\}$. Similarly, the PMEPR of the row sequences is bounded by $\max \{q_{j}^{l_{j}}:1\leq j\leq b\}$.
\end{theorem}
\begin{proof}
    
  \par From the \textit{Theorem 1}, we have 
     $\psi(G_{t})=\{\psi(a_{t}^{\theta}):\theta\in \Theta\}$.
 Let $\left(\psi(a_{t}^{\theta})\right)_{d}$ denotes the $d$-th column of the matrix $\psi(a_{t}^{\theta})$ where $\theta=(r_{1},r_{2},\ldots,r_{a},s_{1},s_{2},\ldots,s_{b})$. From the definition of the complex envelope,  it is easy to observe that for any  sequence $\mathbf{x}$ of length $l$  it is easy to verify that
 \begin{equation}
\abs{S_{\mathbf{x}}(t)}^{2}=l+\sum_{u_{1}\neq0}\boldsymbol{C}(\mathbf{x})(u_{1})\omega_{\lambda}^{-\lambda u_{1}ft}.
 \end{equation}
 Therefore, the complex envelope for the sequence $\left(\psi(a_{t}^{\theta})\right)_{d}$ is given by
 \begin{equation}
 \label{an}
     \abs{S_{\left(\psi(a_{t}^{\theta})\right)_{d}}(t)}^{2}=m+\sum_{u_{1}\neq0}\boldsymbol{C}(\left(\psi(a_{t}^{\theta})\right)_{d})(u_{1})\omega_{\lambda}^{-\lambda u_{1}ft}.
 \end{equation}
 Now, we have 
 \begin{equation}
 \label{a4}
      \boldsymbol{C}[(\psi(a_{t}^{\theta}))_{d}](u_{1})=\sum_{\gamma=0}^{m-1-u_{1}}\omega_{\lambda}^{\left((a_{t}^{\theta})_{\gamma,d}-(a_{t}^{\theta})_{\delta,d}\right)},
 \end{equation}
 where,
  $\delta=\gamma+u_{1}$, $0\leq {\delta},{\gamma}\leq {p_{1}^{m_{1}}p_{2}^{m_{2}}\ldots p_{a}^{m_{a}}-1}$ and $0\leq \abs{u_{1}}\leq p_{1}^{m_{1}}p_{2}^{m_{2}}\ldots p_{a}^{m_{a}}-1$. Let
 \begin{equation}
     \begin{split}
&\gamma=\gamma_{1}+\gamma_{2}p_{1}^{m_{1}}+\gamma_{3}p_{1}^{m_{1}}p_{2}^{m_{2}}+\ldots+\gamma_{a}p_{1}^{m_{1}}p_{2}^{m_{2}}\ldots p_{a-1}^{m_{a-1}},\\ &\delta=\delta_{1}+\delta_{2}p_{1}^{m_{1}}+\delta_{3}p_{1}^{m_{1}}p_{2}^{m_{2}}+\ldots+\delta_{a}p_{1}^{m_{1}}p_{2}^{m_{2}}\ldots p_{a-1}^{m_{a-1}}.
     \end{split}
 \end{equation}
 Substracting $(a_{t}^{\theta})_{\gamma,d}$ and $(a_{t}^{\theta})_{\delta,d}$, we have 
 \begin{equation*}
 \label{maha}
 \begin{split}
    (a_{t}^{\theta})_{\gamma,d}-(a_{t}^{\theta})_{\delta,d}&=f_{\gamma,d}-f_{\delta,d}+\sum_{i=1}^{a}\frac{\lambda}{p_{i}}r_{i}\left(\gamma_{i,\pi_{i}(1)}-\delta_{i,\pi_{i}(1)}\right)\\
     &+\sum_{i=1}^{a}\frac{\lambda}{p_{i}}x_{i}\left(\gamma_{i,\pi_{i}(m_{i})}-\delta_{i,\pi_{i}(m_{i})}\right).
 \end{split}
 \end{equation*}
 Therefore,
 \begin{equation}
 \label{kalas}
     \begin{split}
      &\sum_{\gamma=0}^{m-1-u_{1}}\omega_{\lambda}^{\left((a_{t}^{\theta})_{\gamma,d}-(a_{t}^{\theta})_{\delta,d}\right)}
      =\sum_{\gamma=0}^{m-1-u_{1}}\mathcal{W}_{\gamma}\omega_{p_{i}}^{r_{i}\left(\gamma_{i,\pi_{i}(1)}-\delta_{i,\pi_{i}(1)}\right)},  
     \end{split}
 \end{equation}
 where $$\mathcal{W}_{\gamma}=\omega_{\lambda}^{\left(f_{\gamma,d}-f_{\delta,d}+\sum_{i=1}^{a}\frac{\lambda}{p_{i}}x_{i}\left(\gamma_{i,\pi_{i}(m_{i})}-\delta_{i,\pi_{i}(m_{i})}\right)+\sum_{i'=1,i'\neq i }^{a}\frac{\lambda}{p_{i'}}r_{i'}\left(\gamma_{i',\pi_{i'}(1)}-\delta_{i',\pi_{i'}(1)}\right)\right)}.$$
Consider the subset $\Theta_{1}$ of $\Theta$ as
 \begin{equation}
      \begin{split}
       &\Theta_{1}=\{\theta_{1}=(r_{1},r_{2},\ldots,r_{i-1},r_{i},r_{i+1},\ldots,r_{a},s_{1},s_{2},\ldots,s_{b}):0\leq r_{i}\leq p^{k{i}}_{i}-1\}.
      \end{split}
  \end{equation}
Taking the summation both side of (\ref{kalas}), we obtain
\begin{equation}
    \begin{split}
        &\sum_{\theta_{1}\in \Theta_{1} }\sum_{\gamma=0}^{m-1-u_{1}}\omega_{\lambda}^{\left((a_{t}^{\theta_{1}})_{\gamma,d}-(a_{t}^{\theta_{1}})_{\delta,d}\right)}=
        \sum_{r_{i}=0}^{p^{k_{i}}_{i}-1}\sum_{\gamma=0}^{m-1-u_{1}}\mathcal{W}_{\gamma}\omega_{p_{i}}^{r_{i}\left(\gamma_{i,\pi_{i}(1)}-\delta_{i,\pi_{i}(1)}\right)}\\
        &=\sum_{\gamma=0}^{m-1-u_{1}}\mathcal{W}_{\gamma}\sum_{r_{i}=0}^{p^{k_{i}}_{i}-1}\omega_{p_{i}}^{r_{i}\left(\gamma_{i,\pi_{i}(1)}-\delta_{i,\pi_{i}(1)}\right)}
    \end{split}
\end{equation}
 To derive the PMEPR bound, we consider the two cases as follows.
\begin{mycases}
\case{  $(\gamma_{i,\pi_{i}(1)}\neq\delta_{i,\pi_{i}(1)}$, for any $i\in\{1,2,\ldots,a\})$.}\\
 We have $$\sum_{r_{i}=0}^{p^{k_{i}}_{i}-1}\omega_{p_{i}}^{r_{i}\left(\gamma_{i,\pi_{i}(1)}-\delta_{i,\pi_{i}(1)}\right)}=0.$$ 
 Hence, for such $\gamma$, we have 
 \begin{equation}
 \label{a1}
     \sum_{\theta_{1}\in \Theta_{1} }\sum_{\gamma=0}^{m-1-u_{1}}\omega_{\lambda}^{\left((a_{t}^{\theta_{1}})_{\gamma,d}-(a_{t}^{\theta_{1}})_{\delta,d}\right)}=0.
 \end{equation}
 \case{ ($\gamma_{i\pi_{i}(1)}=\delta_{i\pi_{i}(1)}$ $\forall i\in\{1,2,\ldots,a\}$)}.\\
Just as in the case of proof of Thm \ref{KB},   \textit{Subcase (ii)}  of \textit{Case 1}, it can be shown that
 \begin{equation}
 \begin{split}
&\!\sum_{\theta_{1}\in\Theta_{1}}\!\!\!\!\omega^{\left((a_{t}^{\theta_{1}})_{\gamma,d
     }-(a_{t'}^{\theta_{1}})_{\delta,d}\right)}\!\!+\!\!\!\sum_{\theta_{1}\in \Theta_{1}}\!\!\!\!\!\sum_{k_{1}=1}^{\gamma_{1\pi_{1}(v-1)}}\!\!\!\!\omega_{\lambda}^{\left((a_{t}^{\theta_{1}})_{\gamma^{k_{1}},d}-(a_{t'}^{\theta_{1}})_{\delta^{k_{1}},d}\right)}\!\\
     &+\!\!\!\sum_{\theta_{1}\in \Theta_{1}}\sum_{k_{2}=1+\gamma_{1\pi_{1}(v-1)}}^{p_{1}-1}\!\!\!\!\!\!\!\!\!\!\!\!\omega_{\lambda}^{\left((a_{t}^{\theta_{1}})_{\gamma^{k_{2}},d}-(a_{t'}^{\theta_{1}})_{\delta^{k_{2}},d}\right)}=0.
 \end{split}
 \end{equation}
 \end{mycases}
 Combining the above two cases and (\ref{a4}),   we have 
 \begin{equation}
 \label{kasa}
     \sum_{\theta_{1}\in \Theta_{1}}\boldsymbol{C}\left[(\psi(a_{t}^{\theta_{1}}))_{d}\right](u_{1})=0.
 \end{equation}
 Taking summation on both sides of (\ref{an}), we have 
 \begin{equation}
 \begin{split}
      \sum_{\theta_{1}\in \Theta_{1}}\abs{S_{\left(\psi(a_{t}^{\theta})\right)_{d}}(t)}^{2}
      &=\sum_{\theta_{1}\in \Theta_{1}}m+\sum_{\theta_{1}\in \Theta_{1}}\sum_{u_{1}\neq0}\boldsymbol{C}\left(\left(\psi(a_{t}^{\theta_{1}})\right)_{d}\right)(u_{1})\omega_{\lambda}^{-\lambda u_{1}ft}\\
      &=\sum_{\theta_{1}\in \Theta_{1}}m+\sum_{u_{1}\neq0}\omega_{\lambda}^{-\lambda u_{1}ft}\sum_{\theta_{1}\in \Theta_{1}}\boldsymbol{C}\left(\left(\psi(a_{t}^{\theta_{1}})\right)_{d}\right)(u_{1})\\
      &=mp_{i}^{k_{i}}.
 \end{split}
 \end{equation}
Given that $\theta$ belongs to $\Theta_{1}$, it follows that
\begin{equation}
    \abs{S_{\left(\psi(a_{t}^{\theta})\right)_{d}}(t)}^{2}\leq \sum_{\theta_{1}\in \Theta_{1}}\abs{S_{\left(\psi(a_{t}^{\theta_{1}})\right)_{d}}(t)}^{2}=mp_{i}^{k_{i}}.
\end{equation}
 Hence the column sequence PMEPR is bound by max $\{p_{i}^{k_{i}}:1\leq i\leq a\}$. Similarly, the row sequence PMEPR is also bounded by max $\{q_{j}^{l_{j}}:1\leq j\leq b\}$.
 \end{proof}
\begin{remark}
    Since 2D-CCC is a collection of 2D-GCAS, the proposed work also produces 2D-GCAS with a set size. $\prod_{i=1}^{a}p^{k_{i}}_{i}\prod_{j=1}^{b}q^{l_{j}}_{j}$, row size 
    $\prod_{i=1}^{a}p^{m_{i}}_{i}$ and coloumn size $\prod_{j=1}^{b}q^{n_{j}}_{j}$.
\end{remark}
\begin{remark}
When $a=b=1$, $p_1=q_1= 2$, and $k_1=l_1=1$, \textit{Theorem} 1 provides a 2D-GCAS of array size $2^{m_1}\times 2^{n_1}$. In this case, our row and column sequence PMEPR is bounded by $2$ which is comparable to the 2D-GCAS of array size $2^{n}\times 2^{n}$ in \cite[Th. 16]{pai2022two}. However, in \cite[Th. 16]{pai2022two}, the PMEPRs of the row and column sequences are upper bound by $2^v$, where $v$ is greater than or equal to $1$ and it depends on the choice of permutation functions $\pi_s$. 
Therefore, the proposed 2D-CCC is a promising solution for reducing PMEPR in comparison to \cite{pai2022two}.
\end{remark}
\begin{example}
 In this example, we derive  2D-GCAS from 2D-CCC and look at the row and column sequence PMEPR. We have taken $a=1,b=1,p_{1}=2,q_1=2,m_1=3,k_{1}=1,l_{1}=1$, and $n_1=4$ in \textit{Theorem} 1, and get  2D-GCAS with array size  $8\times 16$. The row and column sequence PMEPRs are computed as 1.8093 and 2, respectively. In contrast, for a similar 2D-CCC in \cite{pai2022designing}, the PMEPR of the row sequence is 3.1237 and that of the column sequence is 2.3621, indicating a notable reduction in our presented approach. In Fig. \ref{fig:my_label}, we have shown the
corresponding plot of row and column sequence IAPR and PMEPR bound.
\end{example}
\begin{figure}[H]
    \centering
    \includegraphics[width=.7 \textwidth]{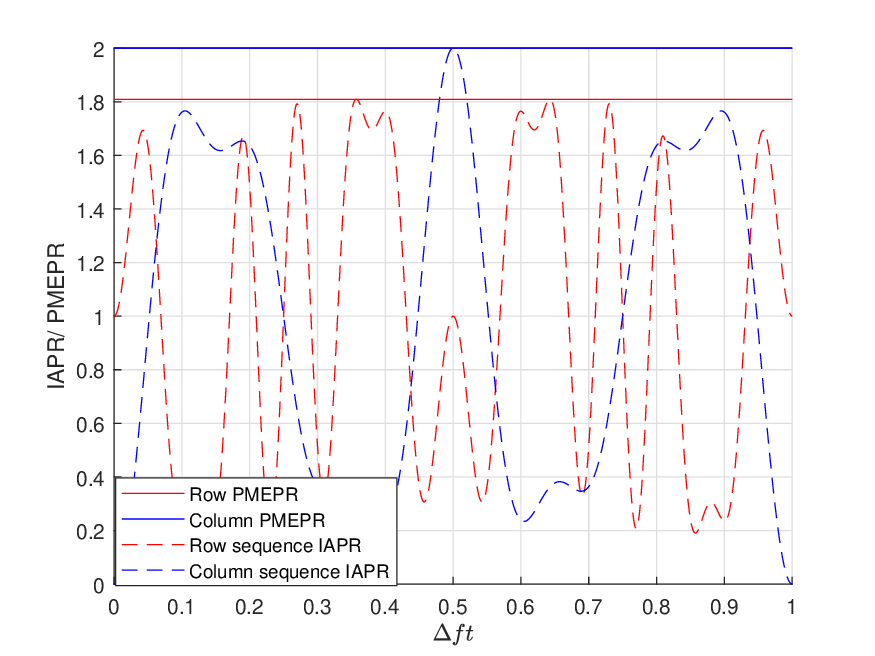}
    \caption{Row and column sequence
IAPR/PMEPR}
    \label{fig:my_label}
\end{figure}
\subsection{Generalization of the proposed work}
This section illustrates the generalizability of the work. 
 \par In \textit{Theorem} \ref{KB} , if we take $a=1,b=1,p_{1}=p,q_{1}=q$, then we have $p^{k_{1}}q^{k_{2}}$ number GCAS with array size $p^{m_{1}}\times q^{n_{1}}$ and set size $p^{k_{1}}q^{k_{2}}$, where $m_{1}>0$ and $n_{1}>0$. Therefore, this paper covers the GCAS construction given in \cite{liu2022constructions}.
\par In \textit{Theorem} \ref{KB} , by considering $m_{i}=2k_{i}$ and $n_{j}=2t_{j}$,  we have 2D $(MK,MK,M^2,K^2)-CCC$. In \cite{farkas2003two}, Farkas \textit{et al}. proposed 2D-CCC  with parameter $(MK,MK,M^2,K^2)-CCC$ and the method depends on the initial choice of 1D sequences. Therefore, the proposed work covers 2D-CCC given in \cite{farkas2003two} without depending on 1D sequences.
\par In \textit{Theorem} \ref{KB} , if we take  $m=1,b=k,q_{1}=q_{2}=\cdots=q_{k}=2$ and $n_{1}+n_{2}+\hdots+n_{k}=n$, then we can cover the CCC given in \cite{chen2008complete,rathinakumar2008complete}.
 \par In \textit{Theorem} \ref{KB} , if we consider $a=1,b=1,p_{1}=2,q_{1}=2,m_{1}=m,n_{1}=n$, then we have $2^{k}$ number of GCAS with array size $2^{m}\times 2^{n}$ and set size $2^{k}$ where $k=k_{1}+t_{1}$ and $m_{1},n_{1},k_{1},t_{1}\geq 1$. Therefore, the GCAS proposed in \cite{pai2022two,pai2021designing} can be constructed by the proposed method except for GCAP.
 \par In \textit{Theorem} \ref{KB}, if we consider $a=1,b=1,p_{1}=p,q_{1}=p,m_{1}=m,n_{1}=n$, then we have $p^{k}$ number of GCAS with array size $p^{m}\times p^{n}$ and set size $p^{k}$, where $k=k_{1}+t_{1}$ and $m_{1},n_{1},k_{1},t_{1}\geq 1$. Therefore, the GCAS proposed in \cite{liu2022constructions} can be constructed by the proposed method except, GCAP.
 In \textit{Theorem} \ref{KB}, by considering $p_{1}=2$, $k_{1}=k$, $m_{1}=kr$ and $n_{1}=ks$,  we have $2^{k}$ GCAS with array size $2^{kr}\times2^{ks}$ and set size $2^{k}$ where $k,r,s$ are positive integer.  Therefore, the GCAS proposed in \nocite{wang2021new} \cite[Th.7]{wang2021new12} can be constructed by the proposed method except $k=1$.
\par In \cite{sarkar2021multivariable}, Sarkar \textit{et al}. constructed 1D-CCC with all possible lengths and set size $\prod_{i=1}^{k}p_{i}^{k_{i}}$. In \textit{Theorem} \ref{KB}, if we take $m=1$, then our proposed construction covers the CCC in \cite{sarkar2021multivariable}. Since CCC is just a group of GCS, our work also makes GCS with any length of sequence.
\section{Comparison with Previous Works}\label{SEC5}
\par We also compare the proposed constructions with previous indirect constructions given in \cite{turcsany2003two,davidekova2013generalized,das2020two} and direct construction given in \cite{pai2021designing} and provide these in TABLE I. 

\begin{table}[h] \centering
\caption{Comparision of proposed 2D-CCC with existing works}
\label{tab:lcsgrammar}
	\scalebox{1}{
\begin{tabular}
{|c|c|c|c|c|}
\hline {Source}  & {\makecell{Set Size}} & {\makecell{Array Size}} & {Based on} & {Constraints} \\
  \hline \cite{turcsany2003two}  & $MK$ & $M^{2}\times K^{2}$ & \makecell{$1 \mathrm{D}-\mathrm{CCCs}$} & {$M, K \geq 2$} \\
   \hline \cite{davidekova2013generalized}  & $M^{2}$ & $2^{m}M\times2^{m}M$ & \makecell{$1 \mathrm{D}-\mathrm{CCCs}$} & {$m \geq 1$,$M\geq 2$} \\
\hline \cite{das2020two} &  $M$ & $M\times M$ & \makecell{BH matrices} & {$M \geq 2$}  \\
\hline \cite{pai2021designing} &  $2^{k}$ & $2^{n}\times2^{m}$ & \makecell{GBF} & {$m,n \geq k$,~$m,n,k\in\mathbb{Z}^{+}$}  \\
\hline Theorem \ref{KB}  & $x$ & $s_{1}\times s_{2}$ & \makecell{MVF} & \makecell{$x=\prod_{i=1}^{a}p^{k_{i}}_{i}\prod_{j=1}^{b}q^{l_{j}}_{j},$\\ 
$s_{1}=\prod_{i=1}^{a}p_{i}^{m_{i}}$,\\
$s_{2}=\prod_{j=1}^{b}q_{j}^{n_{j}}$}\\
\hline
\end{tabular}}
\end{table}

The constructions given in  \cite{turcsany2003two,davidekova2013generalized,das2020two}  depend on the initial choices of some sequence and matrices, whereas the construction given in \cite{pai2021designing} is obtained by GBF. The array size and the set size in \cite{pai2021designing} is of power-of-two. Our construction is based on MVF, which does not need any sequences or matrices initially. Also, it produces all possible array size and flexible set size.

\section{Conclusion}\label{SEC6}
This paper presents new results in the design of 2D-CCC, which is important for omnidirectional transmission in massive MIMO. We have introduced an MVF-based direct 2D-CCC construction. The set size of the proposed 2D-CCC is of the form $\prod_{i=1}^{a}p^{k_{i}}_{i}\prod_{j=1}^{b}q^{l_{j}}_{j}$ which is more flexible and has not been reported in the existing literature. The suggested 2D-CCC can handle all array sizes, which is also not reported.
 As a special case of the design, we have come up with 2D-GCAS with any array size and  flexible set size of the form $\prod_{i=1}^{a}p_{i}\prod_{j=1}^{b}q_{j}$. We have also successfully used the proposed 2D-CCC as precoding matrices in OP-based massive MIMO URA systems. We have shown that the BER performance of the proposed scheme. Moreover, the PMEPR of row and column sequences have been examined. The 2D-CCC given in \cite{farkas2003two}, 
GCP given in \cite[Th.3]{davis1999peak}, CCC given in\cite{chen2008complete,rathinakumar2008complete,sarkar2021multivariable}, and GCAS given in 
\cite[Th.1]{pai2022two} appears as a special case of the proposed construction.

\appendix
\section{Appendix: Proof of \textit{Theorem} \ref{KB}}\label{kn}
\begin{proof}
Let $\gamma,\delta \in \{0,1,\ldots,p_{1}^{m_{1}}p_{2}^{m_{2}}\ldots p_{a}^{m_{a}}-1\}$ and $\mu,\rho \in \{0,1,\ldots,q_{1}^{n_{1}}q_{2}^{n_{2}}\ldots q_{b}^{n_{b}}-1\}$, be such that 
\begin{equation*}
    \begin{split}
     &\gamma=\gamma_{1}+\displaystyle\sum_{i=2}^{a}\gamma_{i}\left(\prod_{i_{1}=1}^{i-1}p_{i_{1}}^{m_{i_{1}}}\right),
     \mu=u_{1}+\displaystyle\sum_{j=2}^{b}\mu_{j}\left(\prod_{j_{1}=1}^{j-1}q_{j_{1}}^{n_{j_{1}}}\right),\\
     &\delta=\delta_{1}+\displaystyle\sum_{i=2}^{a}\delta_{i}\left(\prod_{i_{1}=1}^{i-1}p_{i_{1}}^{m_{i_{1}}}\right),
     \rho=\rho_{1}+\displaystyle\sum_{j=2}^{b}\rho_{j}\left(\prod_{j_{1}=1}^{j-1}q_{j_{1}}^{n_{j_{1}}}\right),
    \end{split}
\end{equation*}
where, $0\leq \gamma_{i},\delta_{i}\leq p_{i}^{m_{i}}-1$ and $0\leq \mu_{j},\rho_{j}\leq q_{j}^{n_{j}}-1$.
  We choose $u_{1}$ and $u_{2}$ such that $\delta=\gamma+u_{1}$, $\rho=\mu+u_{2}$, where $0\leq |u_{1}|\leq p_{1}^{m_{1}}p_{2}^{m_{2}}\ldots p_{a}^{m_{a}}-1$ and $0\leq |u_{2}|< q_{1}^{n_{1}}q_{2}^{n_{2}}\ldots q_{b}^{n_{b}}-1$. The ACCS between $G_{t}$ and $G_{t'}$ can be expressed as
\begin{equation}
\label{g}
    \begin{split}
        \boldsymbol{C}(\psi(G_{t}),\psi(G_{t'}))(u_{1},u_{2})
        &=\sum_{\theta\in \Theta}\boldsymbol{C}\big(\psi(a_{t}^{\theta}),\psi(a_{t'}^{\theta})\big)(u_{1},u_{2})\\
        &=\sum_{\theta\in \Theta}\sum_{\gamma=0}^{m-1-u_{1}}\!\sum_{\mu=0}^{n-1-u_{2}}\omega_{\lambda}^{\left((a_{t}^{\theta})_{\gamma,\mu}-(a_{t'}^{\theta})_{\delta,\rho}\right)}\\
        &=\sum_{\gamma=0}^{m-1-u_{1}}\sum_{\mu=0}^{n-1-u_{2}}\sum_{\theta\in \Theta} \omega_{\lambda}^{\left((a_{t}^{\theta})_{\gamma,\mu}-(a_{t'}^{\theta})_{\delta,\rho}\right)}.
    \end{split}
\end{equation}
By subtracting $(a_{t}^{\theta})_{\gamma,\mu}$ and $(a_{t'}^{\theta})_{\delta,\rho}$, we get
\begin{equation}
\label{f}
    \begin{split}
    &(a_{t}^{\theta})_{\gamma,\mu}-(a_{t'}^{\theta})_{\delta,\rho}
    =f_{\gamma,\mu}-f_{\delta,\rho}+D_{\theta}+D_{\theta}'+E+E',
    \end{split}
\end{equation}
where
\begin{equation*}
\begin{aligned}
 \theta&=(r_{1},r_{2},\ldots,r_{a},s_{1},s_{2},\ldots,s_{b}),\\
    D_{\theta}&=\sum_{i=1}^{a}\frac{\lambda}{p_{i}}r_{i}(\gamma_{i,\pi_{i}(1)}- \delta_{i,\pi_{i}(1)}),\\
    D_{\theta}'&=\sum_{j=1}^{b}\frac{\lambda}{q_{j}}s_{j}(\mu_{j,\sigma_{j}(1)}- \rho_{j,\sigma_{j}(1)}),\\
    E&=\sum_{i=1}^{a}\frac{\lambda}{p_{i}}(\gamma_{i,\pi_{i}(m_{i})}x_{i}-\delta_{i,\pi_{i}(m_{i})}x'_{i})
    \end{aligned}
\end{equation*}
and 
\begin{equation*}
     E'=\sum_{j=1}^{b}\frac{\lambda}{q_{j}}(\mu_{j,\sigma_{j}(n_{j})}y_{i}-\rho_{j,\sigma_{j}(n_{j})}y'_{j}).
\end{equation*}
Therefore,
\begin{equation}
\label{z1}
    \sum_{\theta\in \Theta} w_{\lambda}^{\left((a_{t}^{\theta})_{\gamma,\mu}-(a_{t'}^{\theta})_{\delta,\rho}\right)}=\omega_{\lambda}^{\left(f_{\gamma,\mu}-f_{\delta,\rho}+E+E'\right)}FG,
\end{equation}
where
\begin{equation}
F=\sum_{\theta\in\Theta}\omega_{\lambda}^{D_{\theta}}=\left(\prod_{j=1}^{b}q^{l_{j}}_{j}\right)\prod_{i=1}^{a}\left(\sum_{r_{i}=0}^{p^{k_{i}}-1}\omega_{p_{i}}^{r_{{i}}(\gamma_{i,\pi_{i}(1)}- \delta_{i,\pi_{i}(1)})}\right)
\end{equation}
and
\begin{equation}
G=\sum_{\theta\in\Theta}\omega_{\lambda}^{D'_{\theta}}=\left(\prod_{i=1}^{a}p^{k_{i}}_{i}\right)\prod_{j=1}^{b}\left(\sum_{s_{j}=0}^{q^{l_{j}}_{j}-1}\omega_{q_{j}}^{s_{{j}}(\mu_{j,\sigma_{j}(1)}- \rho_{j,\sigma_{j}(1)})}\right).
\end{equation}
The rest of the proof is split into the following cases
\begin{mycases}
\case
          $(u_{2}>0).$\\  
          \vspace{-0.3cm}\subcase
           $\left(\text{For some}~j\in \{1,2,\ldots,b\}, \mu_{j,\sigma_{j}(1)}\neq \rho_{j,\sigma_{j}(1)}\right)$.\\
           Therefore, we have
\begin{equation}
    \sum_{s_{j}=0}^{q^{l_{j}}_{j}-1}\omega_{q_{j}}^{s_{{j}}(\mu_{j,\sigma_{j}(1)}- \rho_{j,\sigma_{j}(1)})}=0.
\end{equation}
As a result, $G=0$. Hence from (\ref{z1}), we have $$\sum_{\theta\in \Theta} w_{\lambda}^{\left((a_{t}^{\theta})_{\gamma,\mu}-(a_{t'}^{\theta})_{\delta,\rho}\right)}=0.$$ 
 
 \subcase \label{kan}$\left(\textit{For all}~j\in \{1,2,\ldots,b\}, \mu_{j,\sigma_{j}(1)}=\rho_{j,\sigma_{j}(1)}\right)$.\\
 Let us assume that $v$ is the smallest integer for which $\mu_{1,\sigma_{1}(v)}\neq \rho_{1,\sigma_{1}(v)}$.
 Let us consider $\mu^{k_{1}}$ and $\rho^{k_{1}}$ be the integers and corresponding  associated vectors are $\phi\left({\mu}^{k_{1}}\right)=(\boldsymbol{\mu}_{1}^{k_{1}},\boldsymbol{\mu}_{2},\ldots,\boldsymbol{\mu}_{b})$ and $\phi\left({\rho}^{k_{1}}\right)=(\boldsymbol{\rho}_{1}^{k_{1}},\boldsymbol{\rho}_{2},\ldots,\boldsymbol{\rho}_{b})$, respectively, where, $\boldsymbol{\mu_{1}^{k_{1}}}$ and  $\boldsymbol{\rho_{1}^{k_{1}}}$  are the vectors which
differ from $\boldsymbol{\mu_{1}}$ and $\boldsymbol{\rho_{1}}$ only at the position $\sigma_{1}(v-1)$,  i.e.,
\begin{equation}
\begin{split}
     &\boldsymbol\mu_{1}^{k_{1}}=(\mu_{1,1},\mu_{1,2},\ldots,\mu_{1,\sigma_{1}(v-1)-1},\mu_{1,\sigma_{1}(v-1)}-k_{1},\mu_{1,\sigma_{1}(v-1)+1},\ldots,\mu_{1,m_{1}}),\\
      &\boldsymbol\rho_{1}^{k_{1}}=(\rho_{1,1},\rho_{1,2},\ldots,\rho_{1,\sigma_{1}(v-1)-1},\rho_{1,\sigma_{1}(v-1)}-k_{1},\rho_{1,\sigma_{1}(v-1)+1},\ldots,\rho_{1,m_{1}}),
\end{split}
\end{equation}
where $1\leq k_{1}\leq \mu_{1,\sigma_{1}(v-1)}$.
Similarly, let us consider $\mu^{k_{2}}$ and $\rho^{k_{2}}$ be the integers and whose associated vectors are $\phi\left({\mu}^{k_{2}}\right)=(\boldsymbol{\mu}_{1}^{k_{2}},\boldsymbol{\mu}_{2},\ldots,\boldsymbol{\mu}_{b})$ and $\phi\left({\rho}^{k_{2}}\right)=(\boldsymbol{\rho}_{1}^{k_{2}},\boldsymbol{\rho}_{2},\ldots,\boldsymbol{\rho}_{b})$ respectively, where $\boldsymbol\mu_{1}^{k_{2}}$ and  $\boldsymbol\rho_{1}^{k_{2}}$  are the vectors differ from $\boldsymbol\mu_{1}$ and $\boldsymbol\rho_{1}$ only at the position $\sigma_{1}(v-1)$,  i.e.,
\begin{equation}
\begin{split}
     &\boldsymbol\mu_{1}^{k_{2}}=(\mu_{1,1},\mu_{1,2},\ldots,\mu_{1,\sigma_{1}(v-1)-1},\mu_{1,\sigma_{1}(v-1)}-k_{2}+q_{1},\mu_{1,\sigma_{1}(v-1)+1},\ldots,\mu_{1,m_{1}}),\\
      &\boldsymbol\rho_{1}^{k_{2}}=(\rho_{1,1},\rho_{1,2},\ldots,\rho_{1,\sigma_{1}(v-1)-1},,\rho_{1,\sigma_{1}(v-1)}-k_{2}+q_{1},\rho_{1,\sigma_{1}(v-1)+1},\ldots,\rho_{1,m_{1}}),
\end{split}
\end{equation}
where $\mu_{1,\sigma_{1}(v-1)}+1\leq k_{2}\leq q_{1}-1$. As $\delta=\gamma+u_{1}$ and $\rho=\mu+u_{2}$,
 therefore, $(\delta,\rho^{k _{i}})=(\gamma,\mu^{k _{i}})+(u_{1},u_{2})$ for $1\leq i \leq 2$. By subtracting $(a_{t}^{\theta})_{\gamma,\mu^{k_{i}}}$ and $(a_{t'}^{\theta})_{\gamma,\rho^{k_{i}}}$,  we get,
\begin{equation}
\label{h1}
    \begin{split}
    &(a_{t}^{\theta})_{\gamma,\mu^{k_{i}}}-(a_{t'}^{\theta})_{\delta,\rho^{k_{i}}}
    =f_{\gamma,\mu^{k_{i}}}-f_{\delta,\rho^{k_{i}}}+D_{\theta}+D'_{\theta}+E+E'.
    \end{split}
\end{equation}
By subtracting (\ref{f}) and (\ref{h1}), we get
\begin{equation}
\label{l}
    \begin{split}
         &(a_{t}^{\theta})_{\gamma,\mu^{k_{i}}}-(a_{t'}^{\theta})_{\delta,\rho^{k_{i}}}-((a_{t}^{\theta})_{\gamma,\mu}-(a_{t'}^{\theta})_{\delta,\rho})
        =\left(f_{\gamma,\mu^{k_{i}}}-f_{\delta,\rho^{k_{i}}}\right)-(f_{\gamma,\mu}-f_{\delta,\rho}).
    \end{split}
\end{equation}

By subtracting $f_{\gamma,\mu}$ with $f_{\gamma,\mu^{k_{1}}}$  and $f_{\delta,\rho}$ with $f_{\delta,\rho^{k_{1}}}$, we get
\begin{equation}
\label{i}
    \begin{split}
    &f_{\gamma,\mu}-f_{\gamma,\mu^{k_{1}}}=\frac{k_{1}\lambda}{q_{1}}(\mu_{1,\sigma_{1}(v-2)}+\mu_{1,\sigma_{1}(v)})+k_{1}c_{1,\sigma(v-1)},\\
     &f_{\delta,\rho}-f_{\delta,\rho^{k_{1}}}=\frac{k_{1}\lambda}{q_{1}}(\rho_{1,\sigma_{1}(v-2)}+\rho_{1,\sigma_{1}(v)})+k_{1}c_{1,\sigma(v-1)}.
    \end{split}
\end{equation}
 By subtracting $\left(f_{\gamma,\mu^{k_{1}}}-f_{\delta,\rho^{k_{1}}}\right)$ with $\left(f_{\gamma,\mu}-f_{\delta,\rho}\right)$,  we get \begin{equation}
\label{j}
\begin{split}
     \left(f_{\gamma,\mu^{k_{1}}}-f_{\delta,\rho^{k_{1}}}\right)-\left(f_{\gamma,\mu}-f_{\delta,\rho}\right)&=(f_{\gamma,\mu^{k_{1}}}-f_{\gamma,\mu})-(f_{\delta,\rho^{k_{1}}}-f_{\delta,\rho})\\
     &=\frac{k_{1}\lambda}{q_{1}}(\rho_{1,\sigma_{1}(v)}-\mu_{1,\sigma_{1}(v)}).
\end{split}
\end{equation}
From (\ref{l}) and (\ref{j}), we have
\begin{equation}
\label{das}
    \begin{split}
        &\sum_{k_{1}=1}^{\mu_{1,\sigma_{1}(v-1)}}\omega_{\lambda}^{\left((a_{t}^{\theta})_{\gamma,\mu^{k_{1}}}-(a_{t'}^{\theta})_{\delta,\rho^{k_{1}}}\right)-\left((a_{t}^{\theta})_{\gamma,\mu}-(a_{t'}^{\theta})_{\delta,\rho}\right)}\\
        &=\sum_{k_{1}=1}^{\mu_{1,\sigma_{1}(v-1)}}\omega_{q_{1}}^{k_{1}(\rho_{1,\sigma_{1}(v)}-\mu_{1,\sigma_{1}(v)})}.
    \end{split}
\end{equation}
Similarly, by subtracting $f_{\gamma,\mu}$ with $f_{\gamma,\mu^{k_{2}}}$  and $f_{\delta,\rho}$ with $f_{\delta,\rho^{k_{2}}}$, we get
\begin{equation*}
\label{m}
    \begin{split}
    &f_{\gamma,\mu}-f_{\gamma,\mu^{k_{2}}}=\frac{(k_{2}-q_{1})\lambda}{q_{1}}(\mu_{1,\sigma_{1}(v-2)}+\mu_{1,\sigma_{1}(v)})+{\kappa},\\
     &f_{\delta,\rho}-f_{\delta,\rho^{k_{2}}}=\frac{(k_{2}-q_{1})\lambda}{q_{1}}(\rho_{1,\sigma_{1}(v-2)}+\rho_{1,\sigma_{1}(v)})+{\kappa},
    \end{split}
\end{equation*}
where ${\kappa}=(k_{2}-q_{1})c_{1,\sigma(v-1)}$. Again, by subtracting $ \left(f_{\gamma,\mu^{k_{2}}}-f_{\delta,\rho^{k_{2}}}\right)$ with $\left(f_{\gamma,\mu}-f_{\delta,\rho}\right)$,
 we get \begin{equation}
\label{n}
\begin{split}
     \left(f_{\gamma,\mu^{k_{2}}}-f_{\delta,\rho^{k_{2}}}\right)-(f_{\gamma,\mu}-f_{\delta,\rho})&=\left(f_{\gamma,\mu^{k_{2}}}-f_{\gamma,\mu}\right)-\left(f_{\delta,\rho^{k_{2}}}-f_{\delta,\rho}\right)\\
     &=\frac{(k_{2}-q_{1})\lambda}{q_{1}}(\rho_{1,\sigma_{1}(v)}-\mu_{1,\sigma_{1}(v)}).
\end{split}
\end{equation}
From (\ref{l}) and (\ref{n}), we have,
\begin{equation}
\label{narottam}
    \begin{split}
        &\sum_{k_{2}=1+\mu_{1,\sigma_{1}(v-1)}}^{q_{1}-1}\omega_{\lambda}^{\left((a_{t}^{\theta})_{\gamma,\mu^{k_{2}}}-(a_{t'}^{\theta})_{\delta,\rho^{k_{2}}}\right)-\left((a_{t}^{\theta})_{\gamma,\mu}-(a_{t'}^{\theta})_{\delta,\rho}\right)}\\
         &=\sum_{k_{2}=1+\mu_{1,\sigma_{1}(v-1)}}^{q_{1}-1}\omega_{q_{1}}^{k_{2}\left(\rho_{1,\sigma_{1}(v)}-\mu_{1,\sigma_{1}(v)}\right)}.
    \end{split}
\end{equation}
Adding (\ref{narottam}) and (\ref{das}), we have
\begin{equation}
\label{35}
\begin{split}
    &\omega_{\lambda}^{\left((a_{t}^{\theta})_{\gamma,\mu}-(a_{t'}^{\theta})_{ \delta,\rho}\right)}+\!\!\!\sum_{k_{1}=1}^{\mu_{1,\sigma_{1}(v-1)}}\omega_{\lambda}^{\left((a_{t}^{\theta})_{\gamma,\mu^{k_{1}}}-(a_{t'}^{\theta})_{\delta,\rho^{k_{1}}}\right)}\\
&+\sum_{k_{2}=1+\mu_{1,\sigma_{1}(v-1)}}^{q_{1}-1}\!\!\!\!\!\!\!\!\!\!\!\omega_{\lambda}^{\left((a_{t}^{\theta})_{\gamma,\mu^{k_{2}}}-(a_{t'}^{\theta})_{\delta,\rho^{k_{2}}}\right)}=0.
\end{split}
\end{equation}
Since (\ref{35}) holds for all $\theta\in \Theta$, therefore,  taking summation on both side of (\ref{35}), we get
\begin{equation*}
\begin{split}
&\sum_{\theta\in\Theta}\!\!\omega_{\lambda}^{\left((a_{t}^{\theta})_{\gamma,\mu
     }-(a_{t'}^{\theta})_{\delta,\rho}\right)}\!+\!\!\sum_{\theta\in \Theta}\!\!\!\!\!\sum_{k_{1}=1}^{\mu_{1,\sigma_{1}(v-1)}}\omega_{\lambda}^{\left((a_{t}^{\theta})_{\gamma,\mu^{k_{1}}}-(a_{t'}^{\theta})_{\delta,\rho^{k_{1}}}\right)}\\
     &+\sum_{\theta\in \Theta}\sum_{k_{2}=1+\mu_{1,\sigma_{1}(v-1)}}^{q_{1}-1}\!\!\!\!\!\!\!\!\!\!\!\omega_{\lambda}^{\left((a_{t}^{\theta})_{\gamma,\mu^{k_{2}}}-(a_{t'}^{\theta})_{\delta,\rho^{k_{2}}}\right)}=0.
\end{split}
\end{equation*}
Combining the above two subcases, we get $\boldsymbol{C}(\psi(G_{t}),\psi(G_{t'}))(u_{1},u_{2})=0$.
\case ($u_{1}=0,u_{2}=0$).
\subcase ($t\neq t'$).\\
Without loss of generality, let us assume that $y_{j}\neq y'_{j}$. Now, we have 
\begin{equation}
\label{17}
    \begin{split}
         &\sum_{\gamma=0}^{m-1}\sum_{\mu=0}^{n-1} w_{\lambda}^{\left((a_{t}^{\theta})_{\gamma,\mu}-(a_{t'}^{\theta})_{\gamma,\mu}\right)}=AB.
    \end{split}
\end{equation}
where
\begin{equation}
    \begin{split}
        &A=\sum_{\gamma=0}^{m-1}\omega_{\lambda}^{\sum_{i=1}^{a}\frac{\lambda}{p_{i}}\gamma_{i,\pi_{i}(m_{i})}(x_{i}-x'_{i})},\\
        &B=\sum_{\mu=0}^{n-1}\omega_{\lambda}^{\sum_{j=1}^{b}\frac{\lambda}{q_{j}}\mu_{j,\sigma_{j}(n_{j})}(y_{j}-y'_{j})}.
    \end{split}
\end{equation}
We have $\mu=\mu_{1}+\displaystyle\sum_{j=2}^{b}\mu_{j}\left(\prod_{j_{1}=1}^{j-1}q_{j_{1}}^{n_{j_{1}}}\right)$ where $0\leq \mu_{j}\leq q_{j}^{n_{j}}-1$ and the vector representation of $\mu_{j}$ with base $q_{j}$ and length $n_{j}$ is $\boldsymbol{\mu}_{j}$ where $\boldsymbol{\mu}_{j}=(\mu_{j,1},\mu_{j,2},\ldots,\mu_{j,n_{j}})\in \mathbb{A}_{q_{j}}^{n_{j}}$.
We define a mapping $h:\mathbb{A}_{q_{j}^{n_{j}}}\rightarrow \mathbb{A}_{q_{j}}$ given by $h(\mu_{j})=\mu_{j,n_{j}}$. The mapping is well defined beacuse $\boldsymbol{\mu}_{j}$ is the vector representation of $\mu_{j}$, and $\mu_{j,n_{j}}$ is the $n_{j}$th element of the vector  $\boldsymbol{\mu}_{j}$. Now with the help of $h$ and the permutation $\sigma_{j}$ we define a composite mapping $h'=\sigma \circ h:\mathbb{A}_{q_{j}^{n_{j}}}\rightarrow \mathbb{A}_{q_{j}}$ given by $h'(\mu_{j})=\sigma \circ h(\mu_{j})=\sigma(h(\mu_{j}))=\sigma(\mu_{j,n_{j}})=\mu_{j,\sigma_{j}(n_{j})}$. The mapping $h'$ is well defined because both the mapping $h$ and $\sigma_{j}$ are well defined. Let $\mathcal{F}=\{\boldsymbol{\mu}=\left(\mu_{1},\mu_{2},\ldots,\mu_{b}\right):0\leq\mu_{j}\leq q_{j}^{n_{j}}-1\}$ then the sum $B$ can be expressed as
\begin{equation}
    B=\sum_{\left(\mu_{1},\mu_{2},\ldots,\mu_{b}\right)\in \mathcal{F}}\omega_{\lambda}^{\sum_{j=1}^{b}\frac{\lambda}{q_{j}}h'(\mu_{j})(y_{j}-y'_{j})},
\end{equation}
i.e.,
\begin{equation}
\label{18}
\begin{split}
     B&=\sum_{\mu_{1}=0}^{q_{1}^{n_{1}}-1}\sum_{\mu_{2}=0}^{q_{2}^{n_{2}}-1}\ldots\sum_{\mu_{b}=0}^{q_{b}^{n_{b}}-1} \omega_{\lambda}^{\sum_{j=1}^{b}\frac{\lambda}{q_{j}}h'(\mu_{j})(y_{j}-y'_{j})}
     =\prod_{j=1}^{b}\left(\sum_{\mu_{j}=0}^{q_{j}^{n_{j}}-1}\!\!\!\omega_{q_{i}}^{h'(\mu_{j})(y_j-y'_j)}\!\!\right).
\end{split}
\end{equation}
Now, we have
\begin{equation}
\label{19}
    \begin{split}
       \sum_{\mu_{j}=0}^{q_{j}^{n_{j}}-1}\!\!\!\omega_{q_{j}}^{h'(\mu_{j})(y_j-y'_j)}=q_{j}^{n_{j}-1}\!\!\!\!\!\!\!\sum_{\mu_{j,\sigma_{j}(n_{j})}=0}^{q_{j}-1}\!\!\!\omega_{q_{j}}^{\mu_{j,\sigma_{j}(n_{j})}(y_j-y'_j)}=0.
    \end{split}
\end{equation}
Therefore, from (\ref{19}), (\ref{18}) and  (\ref{17}), we have
$$\sum_{\gamma=0}^{m}\sum_{\mu=0}^{n} w_{\lambda}^{\left((a_{t}^{\theta})_{\gamma,\mu}-(a_{t'}^{\theta})_{\gamma,\mu}\right)}=0.$$
\subcase  ($t=t'$).\\
\begin{equation}
 \label{2.d}
    \begin{split}
        \boldsymbol{C}(G_{t},G_{t})(0,0)
      &=  \sum_{\theta\in \Theta}\sum_{\gamma=0}^{p_{1}^{m_1}p_{2}^{m_2}\ldots p_{a}^{m_{a}}-1}\sum_{\mu=0}^{q_{1}^{n_1}q_{2}^{n_2}\ldots q_{b}^{n_{b}}-1}w_{\lambda}^{\left((a_{t}^{\theta})_{\gamma,\mu}-(a_{t}^{\theta})_{\gamma,\mu}\right)}\\
      &=p_{1}^{m_{1}+k_{1}}p_{2}^{m_{2}+k_{2}}\ldots p_{a}^{m_{a}+k_{a}}q_{1}^{n_{1}+l_{1}}q_{2}^{n_{2}+l_{2}}\ldots q_{b}^{n_{b}+l_{b}}.
    \end{split}
\end{equation}
Combining the \textit{Subcase (i)} and \textit{Subcase (ii)}, we have  
\begin{equation}
 \begin{split}
  \boldsymbol{C}(G_{t},G_{t'})(0,0)
  =
\begin{cases}
\prod_{i=1}^{a}p_{i}^{m_{i}+k_{i}}\prod_{j=1}^{b}q_{j}^{n_{j}+l_{j}},t= t':\\
0, t\neq t'.
\end{cases}
\end{split}
\end{equation}

For all the other cases, the proof follows exactly the same as \textit{Case 1.} Hence,
combining all the above cases, we have
 \begin{equation}
 \begin{split}
  \boldsymbol{C}(G_{t},G_{t'})(u_{1},u_{2})
  =
\begin{cases}
\prod_{i=1}^{a}p_{i}^{m_{i}+k_{i}}\prod_{j=1}^{b}q_{j}^{n_{j}+l_{j}},
~~~ (u_{1},u_{2})=(0,0), t=t';\\
0, 
~~~ (u_{1},u_{2})\neq(0,0), t=t' ;\\
0,
~~~~ t\neq t'.
\end{cases}
\end{split}
\end{equation}

\end{mycases}

 Hence, the proof follows.
\end{proof}
\bibliography{sn-bibliography}


\end{document}